\documentclass[letterpaper, 10 pt, conference]{ieeeconf}

\IEEEoverridecommandlockouts
\usepackage{cite}
\usepackage{amsmath,amssymb,amsfonts}
\usepackage{algorithmic}
\usepackage{graphicx}
\usepackage{textcomp}
\usepackage{xcolor}
\def\BibTeX{{\rm B\kern-.05em{\sc i\kern-.025em b}\kern-.08em
    T\kern-.1667em\lower.7ex\hbox{E}\kern-.125emX}}
\usepackage{booktabs}

\newcommand{\cN}{{\cal N}}

\newcommand{\cG}{{\cal G}}

\newcommand{\cE}{{\cal E}}
\newcommand{\cI}{{\cal I}}

\newcommand{\cJ}{{\cal J}}

\newcommand{\bY}{{\bf Y}}

\DeclareMathOperator{\re}{Re}
\DeclareMathOperator{\im}{Im}

\usepackage{bm}

\DeclareFontFamily{U}{mathx}{\hyphenchar\font45}
\DeclareFontShape{U}{mathx}{m}{n}{<-> mathx10}{}
\DeclareSymbolFont{mathx}{U}{mathx}{m}{n}
\DeclareMathAccent{\widebar}{0}{mathx}{"73}

\DeclareSymbolFont{bbold}{U}{bbold}{m}{n}
\DeclareSymbolFontAlphabet{\mathbbold}{bbold}

\newtheorem{proposition}{Proposition}

\newtheorem{lemma}{Lemma}

\newcommand{\dtoprule}{\specialrule{0.6pt}{0pt}{0.4pt}%
            \specialrule{0.6pt}{0pt}{\belowrulesep}%
            }
\newcommand{\dbottomrule}{\specialrule{0.6pt}{0pt}{0.4pt}%
            \specialrule{0.6pt}{0pt}{\belowrulesep}%
            }
  
\begin{document}

\title{A Generalized LinDistFlow Model for Power Flow Analysis}
\author{Jianqiao Huang, Bai Cui, Xinyang Zhou, and Andrey Bernstein
\thanks{This work was authored in part by the National Renewable Energy Laboratory, operated by Alliance for Sustainable Energy, LLC, for the U.S. Department of Energy (DOE) under Contract No. DE-AC36-08GO28308. This work was supported by the Laboratory Directed Research and Development Program at the National Renewable Energy Laboratory. The views expressed in the article do not necessarily represent the views of the DOE or the U.S. Government. The U.S. Government retains and the publisher, by accepting the article for publication, acknowledges that the U.S. Government retains a nonexclusive, paid-up, irrevocable, worldwide license to publish or reproduce the published form of this work, or allow others to do so, for U.S. Government purposes.}
\thanks{J. Huang is with the Department of Electrical and Computer Engineering, Illinois Institute of Technology, Chicago, USA (email: jhuang54@hawk.iit.edu).}
\thanks{B. Cui, X. Zhou, and A. Bernstein are with Power System Engineering Center, National Renewable Energy Laborotary, Golden, CO (emails: \{bai.cui, xinyang.zhou, andrey.bernstein\}@nrel.gov).}
}

\maketitle

\begin{abstract}
This paper proposes a new linear power flow model for distribution system with accurate voltage magnitude estimates. The new model can be seen as a generalization of LinDistFlow model to multiphase distribution system with generic network topology (radial or meshed) around arbitrary linearization point. We have shown that the approximation quality of the proposed model strictly dominates that of the fixed-point linearization (FPL) method, a popular linear power flow model for distribution system analysis, when both are linearized around zero injection point. Numerical examples using standard IEEE test feeders are provided to illustrate the effectiveness of the proposed model as well as the improvement in accuracy over existing methods when linearized around non-zero injection points.
\end{abstract}

\section{Introduction}

Power flow analysis is ubiquitous in power system planning, operation and control. The power flow equations are a set of nonlinear equations relating real and reactive power injections to voltages phasors. These nonlinear equations pose
significant computational challenges for real-time power flow analysis. Moreover, they result in non-convexity of optimization problems, like optimal power flow (OPF), state estimation (SE). Due to the non-convexity, there is no guarantee for an algorithm to converge to the global optimal solution.

To cope with the challenge, linear power flow (LPF) models were proposed. Instead of solving nonconvex algebraic equations, voltages can be obtained with a single set of linear equations for power flow analysis.  Subsequently, OPF with LPF models was proposed. The problem can be solved very efficiently with guaranteed convergence. The solution quality of OPF problem depends on the accuracy of the linear model \cite{ZYang2018}. Moreover, due to superior computational efficiency, accurate linear models have been adopted in various power system studies including online control strategies, probabilistic power flow, contingency analysis, reliability evaluation, placement and sizing of inverter-based distributed energy resources (DER) and hosting capacity.


There is a rich literature on LPF models in distribution system. The linearized distribution flow (LinDistFlow) model proposed in \cite{M.E.Baran1989} is arguably the most widely used linear power flow model for distribution system analysis. Compared to transmission system LPF models, LinDistFlow yields better voltage magnitude estimate under radial topology and high r/x ratio \cite{S.C.Tripathy1982,JR2013,K.Purchala2005}. The model approximates squared nodal voltage magnitudes as a linear function of approximate line flow by dropping the quadratic terms relating to branch losses in the nonlinear power flow equations of a single phase radial network. Exploiting the radiality of distribution networks, the equivalent model in \cite{farivar2013equilibrium} is derived, which establishes an approximate linear relationship between squared nodal voltage magnitudes and nodal power injections. Subsequently, LinDistFlow model is extended to radial multiphase distribution systems \cite{LGan2016gradientOPF, V.Kekatos2016, B.A.Robbins2016,R.R.Jha2020,D.Arnold2016}. 
Reference \cite{LGan2016gradientOPF} proposes a linear model by ignoring line losses with the assumption that three phase voltages are balanced. The work in \cite{B.A.Robbins2016,R.R.Jha2020} extends the model in \cite{LGan2016gradientOPF} to recover line losses by linearizing square magnitude of line power flow and line current respectively as a function of line flow around a given operating point. These models are used in designing online OPF solvers or voltage regulation algorithms. 

With the increasing popularity of microgrids and modern protection designs, the future distribution system can be expected to go beyond radial operation paradigm. However, there are no extensions of LinDistFlow model to meshed network to the best of our knowledge. Specifically, there are no explicit approximate linear relationship between squared nodal voltage magnitudes and nodal power injections for networks with general topology which reduces to LinDistFlow model when the network is radial. We try to fill this gap in the present paper.

There are extensive literature on LPF other than LinDistFlow. With regression and optimization techniques, \cite{JR2013, Ahmadi2016ALP, E.Schweitzer2020, S.Misra2018, Mhlpfordt2019OptimalAP} obtain LPF models over a range of operating points. For example, \cite{JR2013} fits load models over a predefined set of operating points to reformulate power flow equations as a LPF. The work in \cite{A.Garces2016, A.B2018} propose LPF models based on first order Taylor (FOT) approach. In addition, \cite{A.B2018} proposes a fixed-point linearization (FPL) method, and numerically shows that FOT method is a better local linear approximator, but FPL method provides a better global approximation. Compared with FOT method, the approximation error of FPL method increases slowly when the exact operating point is far away from the given power flow solution.
Although FPL method has a better global behavior, the approximation error of the FPL method still increases quickly when the exact operating point is moving far away from the linearized point. 


To have a more versatile LPF model with low computational complexity while preserving the merits of LinDistFlow, we propose a new LPF model, which we call generalized LinDistFlow model (GLDF). As the name suggests, GLDF can be seen as a generalization of LinDistFlow model. The contributions of this paper are summarized as follows. First, we show the proposed GLDF model generalizes the LinDistFlow model to multiphase, generic network with meshed or radial topology and arbitrary linearization point. In particular, for a radial network, the proposed model linearized around zero-injection point coincides with single phase and multiphase LinDistFlow model in \cite{M.E.Baran1989,LGan2016gradientOPF}.
Second, we show that the proposed model achieves more accurate voltage estimate than FPL method when they are linearized around zero injection point. When they are linearized around other points, numerical evaluations show that the proposed model achieves a better global behavior than FPL method for a radial distribution system.


\section{System Model}\label{sec: LPF model}
%

\subsection{Notations}
In this paper, we use bold letters to represent matrices, e.g., $\mathbf{A}$, italic bold letters to represent vectors, e.g, $\bm{A}$ and $\bm{a}$, and non-bold letters to represent {scalars}, e.g., $A$ and $a$. For matrix $\mathbf{A}$, $\mathbf{A}^{\top}$, $\widebar{\mathbf{A}}$, $\mathbf{A}^{H}$, $\mathbf{A}^{-1}$ denote its transpose, conjugate, conjugate transpose, and inverse (only for square matrix), respectively. $\mathbf{diag}(\bm{a})$ denotes a diagonal matrix with diagonal $\bm{a}$. $\mathbf{A}_i$ denotes the $i$th column vector of matrix $\mathbf{A}$; $\mathbf{A}_{ij}$ denotes the element that sits in $i$th row and $j$th column of matrix $\mathbf{A}$. When $i$, $j$ are index sets, $\mathbf{A} (i, j)$ is a submatrix by taking the rows in the set $i$ and columns in the set $j$ of $\mathbf{A}$. $\bm{A}_i$ denotes the $i$th element of vector $\bm{A}$. $\mathfrak{i}:=\sqrt{-1}$ is used as the imaginary unit. $\re(\cdot)$ and $\im(\cdot)$ denote the real and imaginary parts of a complex number, vector or a matrix. $|\cdot|$ denotes the cardinality of a set.

\subsection{Distribution System}
We consider a multiphase distribution system with a generic topology, which can be radial or meshed, denote by a graph $\cG=\{\{0\}\cup\cN, \cE\}$, where $\{0\}$, $\cN=\{1, 2, ..., N\}$, $\cE = \{ (i, j) \subseteq \cN \times \cN \}$ denote the set of slack bus, the set of PQ buses and the set of lines between buses, respectively. We use $\cN_{i}=\{j \mid (i,j) \subseteq \cE \}$ to denote the set of adjacent buses of bus~$i$. Set~$\Phi_i$ denotes the available phases of bus~$i$, e.g., we set $\Phi_i=\{a,b,c\}$ for a three-phase bus~$i$. We define each phase of a bus as a node, and define the index set of all the nodes as $\cI:= \{1, ..., n \}$, with the total number of nodes in the system $n=\sum_{i\in\cN} |\Phi_i|$.
Define $\cI_j = \{ i \mid \mathrm{bus}[i] = j \}, j\in\cN$ as the node index set of bus~$j$, where $\mathrm{bus}[i]$ denotes the bus to which node $i$ belongs. Let $\cI_j^{\phi}$, for $\phi \in \Phi_j, j\in\cN$ be the index of bus~$j$ phase~$\phi$.

Similarly, a line may have up to three available phases, each defined as a branch. Let $(j,i) \subseteq \cE$ denote line $i$, where node $j$ is the parent of node $i$ in a radial network. Let $\Omega_i$ denote the available phases of line~$i\in\cE$, e.g., $\Omega_i=\{a,b,c\}$ for a three-phase line~$i$ implies that line~$i$ has three phases and three branches. We define the index set of all the branches $\cJ:= \{1, ..., m \}$, where $m=\sum_{i\in\cE} |\Omega_i|$. 
Let $\cJ_j = \{ i\mid \mathrm{line}[i] = j \}$ be the index set of line $j$, where $\mathrm{line}[i]$ denotes the line to which branch $i$ belongs. Denote by $\bm{y}_k\in\mathbb{C}^{|\Omega_k|\times |\Omega_k|}$ the admittance matrix for line $k$ . If the system is a radial system, $\cG$ becomes a tree. Path is a sequence of edges which joins a sequence of distinct vertices. Denote by $\cE_i \subseteq \cE$ the set of lines forming the unique path from slack bus to bus~$i$. 

\subsection{Power Flow Model}
The voltage phasor of the slack bus $\bm{V}_S$ are fixed and given. For bus~$i\in\cN$, let $\bm{V}_i$, $\bm{I}_i$,  $\bm{S}_i=\bm{p}_i+\mathfrak{i}\bm{q}_i$ be its complex voltage, current injection and power injection. The admittance and impedance matrix for the line between bus $i$ and $j$ are denoted by $\bm{y}_{ij}$ and $\bm{z}_{ij}=\bm{r}_{ij}+\mathfrak{i}\bm{x}_{ij}$, respectively. We assume all resistance and reactance are nonnegative. The bus admittance matrix $\bY$ is obtained by:
\begin{align}
\bY(\cI_i, \cI_j)
=\begin{cases}
    \bm{0}, &\text{if } j \neq i, \text{and } j \not \in \cN_i, \nonumber\\
    -\bm{y}_{ij}, &\text{if } j \in  \cN_i, \nonumber\\ 
    \sum_{g \in \cN_i} \bm{y}_{ig}, & i=j. 
\end{cases}
\end{align}
More details of the multiphase distribution system modeling can be found in \cite{Bazrafshan2017}.  

Given the bus admittance matrix $\bY$, the Ohm's and Kirchhoff's laws relate bus currents and voltages as follows:
\begin{equation} \label{eq:IV}
    \begin{bmatrix} \bm{I}_L\\ \bm{I}_S \end{bmatrix} = 
    \begin{bmatrix}
        \bY_{LL} & \bY_{LS}\\
        \bY_{SL} & \bY_{SS}
    \end{bmatrix}
    \begin{bmatrix} \bm{V}_L \\ \bm{V}_S \end{bmatrix},
\end{equation}
where we partition the buses into the slack bus and the PQ buses, which are signified by subscripts $S$ and $L$, respectively. We denote the constant open-circuit voltage as $\bm{E} := -\bY^{-1}_{LL}\bY_{LS}\bm{V}_S$ and the impedance matrix excluding the slack bus as $\mathbf{Z} := \bY^{-1}_{LL}$. For notational simplicity, we drop the subscripts of $\bm{V}_L$ and $\bm{I}_L$ for the PQ buses, and solve the PQ bus voltages as:
\begin{equation} \label{eq:VL1}
    \bm{V} = \bm{E} + \mathbf{Z}\bm{I}.
\end{equation}
Note that $\bm{I} =\mathrm{diag}(\widebar{\bm{V}})^{-1} \widebar{\bm{S}}$, and we can rewrite \eqref{eq:VL1} as:
\begin{equation} \label{eq:VL2}
    \bm{V} = \bm{E} + \mathbf{Z}\cdot \mathrm{diag}(\widebar{\bm{V}})^{-1}\widebar{\bm{S}}.
\end{equation}

\section{Proposed Linear Power Flow Model}
\label{sec:linearization}

In this section, we propose a new linear power flow model to characterize the relationship between the squared voltage magnitudes and the bus power injections, and discuss how the proposed model can be seen as a generalization of the celebrated LinDistFlow model \cite{M.E.Baran1989}.

\subsection{Linear Power Flow Model}
We denote the squared voltage magnitude at node $i\in\cI$ as $v_i$ and its vectorized form as $\bm{v}\in\mathbb{R}^n$. We can write $v_i$ based on \eqref{eq:VL2} as:
\begin{eqnarray} \label{eq:v1}
    && \hspace{-8mm}v_i = |E_i|^2 + 2\re\Big( E_i \sum_{k=1}^n \widebar{Z}_{ik}\frac{S_k}{V_k} \Big) + \sum_{k = 1}^n \sum_{\ell = 1}^n Z_{ik}\widebar{Z}_{i\ell} \frac{\widebar{S}_k S_\ell}{\widebar{V}_k V_\ell}.\nonumber\\
\end{eqnarray}
To derive a linear approximation between the squared voltage magnitudes and bus power injections, we take the two following technical steps for Eq.~\eqref{eq:v1}: 1) we drop the dependence of $v_i$ on $\bm{V}$ in the denominators of the second and third terms, and 2) we linearize the cross products of the power injections and its conjugate in the third term.

Given the open-circuit voltage $\bm{E}$, a linear approximation of \eqref{eq:v1} around a certain operating point $(\bm{S}^*, \bm{V}^*)$ can be obtained as follows. we first substitute $\bm{E}$ for $\bm{V}$ on the RHS of \eqref{eq:v1} to drop the dependence on variable $\bm{V}$. Next, we replace the third term with $\mathbf{\Lambda}$ in \eqref{eq: Lambda_GLDF}, the difference between the squared magnitudes of $\bm{V}^*$ and the linear approximate evaluation of first two terms at $(\bm{S}^*, \bm{V}^*)$. 
Consequently, with some algebraic manipulations, the following vectorized linear power flow approximation of \eqref{eq:v1} at $(\bm{S}^*, \bm{V}^*)$ can be obtained:
\begin{equation} \label{eq:LPF}
    \boxed{
    \bm{v} = \mathrm{diag}(\bar{\bm{E}})\bm{E} + \mathbf{M} \bm{p} + \mathbf{N} \bm{q} + \mathbf{\Lambda},
    }
\end{equation}
where
\begin{subequations}\label{eq:GLDF_coefficient}
\begin{align}
    \mathbf{M} &= \hspace{2.5mm} 2\re\left( \mathrm{diag}(\bm{E})\widebar{\mathbf{Z}} \cdot \mathrm{diag}(\bm{E})^{-1} \right),\\
    \mathbf{N} &= -2\im\left( \mathrm{diag}(\bm{E})\widebar{\mathbf{Z}} \cdot \mathrm{diag}(\bm{E})^{-1} \right),\\
    \bm{\Lambda} &= \mathrm{diag}(\widebar{\bm{V}}^*)\bm{V}^*-(\mathrm{diag}(\bar{\bm{E}})\bm{E} + \mathbf{M} \bm{p^*} + \mathbf{N} \bm{q}^*) \label{eq: Lambda_GLDF}. 
\end{align}
\end{subequations}
Notice that the matrices $\mathbf{M}, \mathbf{N}$ are constant matrices independent of $\bm{S}^*$ or $\bm{V}^*$, while the vector $\bm{\Lambda}$ depends on $(\bm{S}^*, \bm{V}^*)$. By setting the constant term $\mathbf{\Lambda}$ according to \eqref{eq: Lambda_GLDF}, we guarantee that the proposed linearized model achieves accurate modeling at the linearization point $(\bm{S}^*, \bm{V}^*)$.

\subsection{Recovering LinDistFlow Model as a Special Case}\label{sec:LDF}

In ths part, we show that the proposed model coincides with multiphase LinDistFlow model \cite{LGan2016gradientOPF} when linearized around zero-injection point in radial distribution networks. We assume $\mathcal{G}$ is radial for the rest of the section.

\subsubsection{Multiphase LinDistFlow}
LinDistFlow model, as well as its multiphase system extension, has been widely adopted for power flow linearization in radial distribution networks because of its decent accuracy around zero-loading operation points and its meaningful physical interpretation with respect to topology and line impedance. The multiphase LinDistFlow model is given as follows:
\begin{equation} \label{eq:LinDistFlow}
    \bm{v} = \mathrm{diag}(\bar{\bm{E}})\bm{E} + \widehat{\mathbf{R}} \bm{p} + \widehat{\mathbf{X}} \bm{q},    
\end{equation}
where the sensitivity matrices are calculated as:
\begin{subequations}
\begin{eqnarray}
        \widehat{\mathbf{R}} (\cI_i^{\phi}, \cI_j^{\varphi}) &=&  2\hspace{-4mm}\sum_{(h,k) \subseteq \cE_i \cap \cE_j}\hspace{-4mm}{\re\Big(\alpha^{\phi-\varphi} \bar{z}^{\phi \varphi}_{hk}\Big)}, \\
        \widehat{\mathbf{X}} (\cI_i^{\phi}, \cI_j^{\varphi}) &=&\hspace{-2.5mm} -2\hspace{-4mm}\sum_{(h,k) \subseteq \cE_i \cap \cE_j}\hspace{-4mm}{\im\Big(\alpha^{\phi-\varphi} \bar{z}^{\phi \varphi}_{hk}\Big)},
\end{eqnarray}
\end{subequations}
for any $\phi \in \Phi_i$, $\varphi \in \Phi_j$, $i,j \in \cN$, with $\alpha=e^{-\mathfrak{i} 2\pi/3}$. 
We use $a=0$, $b=1$, and $c=2$ to calculate the phase difference $\alpha^{\phi-\varphi}$. $z^{\phi \varphi}_{hk} \in \mathbb{C}$ is the (mutual) impedance of line $(h, k)$ between phase $\phi$ and $\varphi$. 

We now write down the proposed model \eqref{eq:LPF} linearized at zero injection point. We have $\bm{S}^* = \mathbbold{0}$ and $\bm{V}^* = \bm{E}$ at no load condition. It follows that in this case $\bm{\Lambda} = \mathbbold{0}$. So model \eqref{eq:LPF} simplifies to
\begin{equation} \label{eq:LPF0}
    \bm{v} = \mathrm{diag}(\bar{\bm{E}})\bm{E} + \mathbf{M} \bm{p} + \mathbf{N} \bm{q}. 
\end{equation}

To show that \eqref{eq:LinDistFlow} and \eqref{eq:LPF0} coincide, it suffices to show $\widehat{\mathbf{R}}=\mathbf{M}$, $\widehat{\mathbf{X}}=\mathbf{N}$. From \eqref{eq:GLDF_coefficient}, we know that 
\begin{subequations}
\begin{align}
   \mathbf{M} (\cI_i^{\phi}, \cI_j^{\varphi}) &= \hspace{2.5mm}
   2\re \Big(\alpha^{\phi-\varphi} \widebar{\mathbf{Z}} (\cI_i^{\phi}, \cI_j^{\varphi} )\Big),\\ 
   \mathbf{N} (\cI_i^{\phi}, \cI_j^{\varphi}) &= 
   -2\im\Big(\alpha^{\phi-\varphi} \widebar{\mathbf{Z}} (\cI_i^{\phi}, \cI_j^{\varphi} )\Big). 
\end{align}
\end{subequations}
Therefore, we only need to prove that:
\begin{align}
    \mathbf{Z} (\cI_i, \cI_j)=\hspace{-4mm}\sum_{(h,k) \subseteq \cE_i \cap \cE_j}\hspace{-4mm}{\bm{z}_{hk}}, \label{eq:commonpath}
\end{align}
for any $i, j \in \cN$. 

\subsubsection{Reduced Oriented Incidence Matrix}
To this end, we first define a multiphase line-to-bus reduced oriented incidence matrix $\mathbf{A}$. The complete oriented incidence matrix encodes the information on connectivity and line orientation. The reduced one is obtained by eliminating the columns corresponding to slack bus. More precisely, the block $\mathbf{A} (\cJ_k, \cI_i)$ is a $|\cJ_k|$-by-$|\cI_i|$ matrix defined as:
\begin{eqnarray}\label{eq:defA}
\mathbf{A} (\cJ_k^{\varphi}, \cI_i^{\phi})
=\begin{cases}
    0, &\text{if } \phi \neq \varphi \text{ or line } k \text{ is not incident } \\
    &\text{to bus } i;\\
    -1, &\text{if } \phi=\varphi \text{ and line } k \text{ is directed} \\ 
    &\text{  toward bus } i,\\
    1, &\text{if } \phi=\varphi \text{ and line } k \text{ is directed} \\ 
    &\text{  out of bus } i,
\end{cases}\hspace{-3mm}
\end{eqnarray}
where $\varphi \in \Omega_k$, $\phi \in \Phi_i$. 

For example, in a 3-bus system where $\cE=\{ (0, 1), (1,2)\}$, Bus~$0$ is the slack bus, $(0, 1)$, $(1, 2)$ denote line~$1$ and line~$2$, respectively, $\Phi_1 = \{a, b, c \}$, $\Phi_2 = \{ a, c\}$ and lines are directed toward the slack bus. The submatrices of the reduced incident matrix are,
\begin{subequations}
\begin{align}
    \mathbf{A} (\cJ_1, \cI_1) &= 
    \begin{bmatrix}
        1 & 0 & 0\\
        0 & 1 & 0\\
        0 & 0 & 1
    \end{bmatrix}, \quad 
    \mathbf{A} (\cJ_1, \cI_2) = 
    \begin{bmatrix}
        0 & 0 \\
        0 & 0 \\
        0 & 0
    \end{bmatrix},\nonumber\\
    \mathbf{A} (\cJ_2, \cI_1) &=-
    \begin{bmatrix}
        1 & 0 & 0\\
        0 & 0& 1
    \end{bmatrix}, \quad
    \mathbf{A} (\cJ_2, \cI_2) =
    \begin{bmatrix}
        1 & 0\\
        0 & 1
    \end{bmatrix}.\nonumber
\end{align}
\end{subequations}

We next present the form of the inverse of matrix $\mathbf{A}$.

\begin{lemma}\label{lem:inverseA}
    The inverse of matrix $\mathbf{A}$ is characterized as:
    \begin{align}\label{eq:Ainverse}
    \mathbf{A}^{-1} (\cI_i^{\phi}, \cJ_k^{\varphi})
    =\begin{cases}
         0, &\text{if } \phi \neq \varphi \text{ or line } k \not \in \cE_i;\\
        -1, \hspace{-2mm}&\text{if } \phi=\varphi \cap \{ \text{line } $k$ \text{ is directed against}\\ 
        &\text{  path from bus } $i$ \text{ to slack bus} \};\\
        1, &\text{if } \phi=\varphi \cap \{ \text{line } $k$ \text{ is directed along}\\
        &\text{  path from bus } $i$ \text{ to slack bus} \}.
    \end{cases}
    \end{align}
\end{lemma}
\begin{proof}
    We first obtain a permuted $\mathbf{A}$ by
    \begin{eqnarray} \label{eq:simtrans}
    \widehat{\mathbf{A}} &=& \mathbf{P} \mathbf{A} \mathbf{P}^{\top}\label{eq:permuteA}
    \end{eqnarray}
    such that 
        $$\widehat{\mathbf{A}}=     
        \begin{bmatrix}
        \widehat{\mathbf{A}}_{aa} & \widehat{\mathbf{A}}_{ab} & \widehat{\mathbf{A}}_{ac}\\
        \widehat{\mathbf{A}}_{ba} & \widehat{\mathbf{A}}_{bb} & \widehat{\mathbf{A}}_{bc}\\ 
        \widehat{\mathbf{A}}_{ca} & \widehat{\mathbf{A}}_{cb} & \widehat{\mathbf{A}}_{cc}\\ 
        \end{bmatrix},$$ 
        where $\widehat{\mathbf{A}}_{\phi\varphi} \in  \mathbb{R}^{|\cE^{\phi}| \times |\cN^{\varphi}|}$ for any $\phi, \varphi \in \{ a, b, c\}$ denotes the reduced incident matrix for phase $\phi$ of edges and phase $\varphi$ of buses, with $\cE^{\phi}$ and $\cN^{\varphi}$ collecting all indices of edges and buses having phase $\phi$ and $\varphi$, respectively. Matrix $\mathbf{P}$ is the permutation matrix.
    
    According to the definition of $\mathbf{A}$ in  \eqref{eq:defA}, the off-diagonal blocks of $\widehat{\mathbf{A}}$ are all zeros, i.e.,
    $\widehat{\mathbf{A}}=\mathrm{blkdiag} (\widehat{\mathbf{A}}_{aa}, \widehat{\mathbf{A}}_{bb}, \widehat{\mathbf{A}}_{cc})$. 
    
    
    Since the similarity transformation \eqref{eq:simtrans} preserves rank of $\mathbf{A}$, $\widehat{\mathbf{A}}$ is invertible as long as $\mathbf{A}$ is invertible. Due to the block diagonal structure, $\widehat{\mathbf{A}}$ is invertible if and only if all diagonal blocks are invertible. Note that each of the diagonal blocks is a reduced oriented incidence matrix of a tree. We know that the incidence matrix of a connected graph is full rank \cite[Thm. 8.3.1]{godsil2001algebraic}, so a reduced incidence matrix is invertible as long as it is a square matrix. It follows that all diagonal blocks of $\widehat{\mathbf{A}}$ are invertible, so both $\widehat{\mathbf{A}}$ and $\mathbf{A}$ are invertible.
    
    By simultaneously multiplying $\mathbf{P}^{-1}$ and $({\mathbf{P}^{\top}})^{-1}$ on both LHS and RHS of \eqref{eq:permuteA} and taking inverse on both sides, we obtain the inverse of $\mathbf{A}$ as:
    \begin{align}
        \mathbf{A}^{-1}=\mathbf{P}^{\top} \widehat{\mathbf{A}}^{-1} \mathbf{P}.
    \end{align}
    Based on the property of permutation matrix,  $\mathbf{P}^{-1}=\mathbf{P}^{\top}$, then,
    \begin{align}
        \mathbf{A}^{-1}= \mathbf{P}^{-1}  \widehat{\mathbf{A}}^{-1} (\mathbf{P}^{\top})^{-1}. \label{eq: inverse_A}
    \end{align} 
    Here, matrix $\widehat{\mathbf{A}}^{-1}$ also has a block diagonal structure, each block of which can be independently calculated according to the inverse of the single-phase reduced incident matrix in \cite{LCPF1} as:
    \begin{align}\label{eq:hatAinverse}
    \widehat{\mathbf{A}}_{\phi\phi}^{-1} (i, k)
    =\begin{cases}
         0, &\text{if line } k \not \in \cE_{i};\\
        -1, \hspace{-2mm}&\text{if line } k \text{ is directed against}\\ 
        &\text{  path from bus } i \text{ to slack bus};\\
        1, &\text{if line } k \text{ is directed along}\\
        &\text{  path from bus } i \text{ to slack bus},
    \end{cases}
    \end{align}
for any $\phi\in\{a,b,c\}$.

Note that the permutation in \eqref{eq: inverse_A} reverses that in \eqref{eq:permuteA}. Therefore, $\mathbf{A}^{-1}$ is $\widehat{\mathbf{A}}^{-1}$ permuted back in the order of the original $\mathbf{A}$, and \eqref{eq:Ainverse} follows.
\end{proof}

\subsubsection{Equivalence Proof} We now show Eq.~\eqref{eq:commonpath} holds by the following proposition.

\begin{proposition}\label{lem:commonpath}
	Given a radial system, for any two buses $i, j \in \cN$, $\bY^{-1}_{LL} (\cI_i, \cI_j)$ is the sum of impedance matrices of lines in the set $\cE_i \cap \cE_j$, i.e., the summarized impedance of the common path of buses $i$ and $j$ leading back to the slack bus.
\end{proposition} 
\begin{proof}
For a single phase radial distribution system, Proposition~\ref{lem:commonpath} has been proved in \cite{LCPF1}. Here we prove the multiphase case. We first express $\bY_{LL}$ in terms of the reduced incidence matrix and line admittance matrix as:
\begin{align}
    \bY_{LL}&=\mathbf{A}^{\top}\mathbf{y}\mathbf{A},\label{eq:YLL}
\end{align}
where $\mathbf{y}\in \mathbb{C}^{m \times m}$ is a block diagonal matrix of line admittance and $\mathbf{A}$ is the multiphase line to bus reduced directed incidence matrix defined previously.

Recall that $\mathbf{Z}=\bY^{-1}_{LL}$. Then based on \eqref{eq:YLL}, we have the following:
\begin{align}
    \mathbf{Z} &=\mathbf{A}^{-1} \mathbf{y}^{-1} (\mathbf{A}^{\top})^{-1}\nonumber\\
                  &=\sum_{k\in\cE} (\mathbf{A}^{-1})_{\cJ_k} \bm{z}_k
                  \left((\mathbf{A}^{-1})_{\cJ_k}\right)^{\top},\label{eq:H_1r_exp}
\end{align}
where $\bm{z}_k \in \mathbb{C}^{|\Omega_k| \times |\Omega_k|}$ is the impedance matrix of line~$k$. $(\mathbf{A}^{-1})_{\cJ_k}$ is the $|\Omega_k|$ columns of $\mathbf{A}^{-1}$ corresponding to the $k$th line. Then, we can obtain the following for any buses $i$ and $j$:
\begin{align}
   \mathbf{Z} (\cI_i, \cI_j) =
   \sum_{k\in\cE} \mathbf{A}^{-1} (\cI_i, \cJ_k) \bm{z}_k
                  \left(\mathbf{A}^{-1} (\cI_j, \cJ_k)\right)^{\top},
\end{align}
where $\mathbf{A}^{-1} (\cI_i, \cJ_k)$ is an $|\cI_i|$-by-$|\cJ_k|$ submatrix of $\mathbf{A}^{-1}$ defined by Eq.~\eqref{eq:Ainverse} in Lemma~\ref{lem:inverseA}.

Therefore, $ \mathbf{Z} (\cI_i, \cI_j)$ is the sum of  $|\Phi_{i}|$-by-$|\Phi_{j}|$ impedance matrix of lines in the set $\cE_i \cap \cE_j$, where $i, j \in \cN$. 
\end{proof}

Based on Proposition~\ref{lem:commonpath}, \eqref{eq:LPF0} is equivalent to the multiphase LinDistFlow model \eqref{eq:LinDistFlow} in \cite{LGan2016gradientOPF}. For a single phase distribution system, \eqref{eq:LPF0} is equivalent to the LinDistFlow model in \cite{M.E.Baran1989, farivar2013equilibrium}. Therefore, the proposed LPF model \eqref{eq:LPF} can be seen as a generalization of LinDistFlow.


\section{Error Analysis of the Proposed Model} \label{sec:Error Analysis}



Bus injection model and branch flow model are two equivalent models commonly used for power flow analysis \cite{B.Subhonmesh2012}. Both models have been used to derive linear power flow models. The proposed model \eqref{eq:LPF} falls into the first category since the approximation is based on \eqref{eq:VL2} which involves exclusively nodal variables. On the other hand, LinDistFlow is arguably the most widely used branch flow model-based linear approximation. Having established the connection between the proposed model and LinDistFlow in the previous section, we now investigate how it compares to bus injection-based linear power flow models. 

The fixed-point power flow formulation \eqref{eq:VL2} is popular for distribution system analysis. It is not used as often in transmission system due to its inability to model PV bus. As a result, many bus injection-based linear models tailored for distribution system are based on \eqref{eq:VL2}, such as \cite{A.B2018}, \cite{Bolognani2015}, and \cite{Dhople2015}, as opposed to ones for transmission system that are generally derived from the admittance matrix formulation \eqref{eq:IV}. An example of distribution system linear power flow model based on \eqref{eq:VL2} is the FPL model recently proposed in \cite{A.B2018, A.Bernstein2017_conf}. This model linearizes around an operating point $\bm{V}^*$ by simply replacing $\bm{V}$ in \eqref{eq:VL2} by $\bm{V}^*$, which writes:
\begin{equation} \label{eq:FPL}
    \bm{V} = \bm{E} + \mathbf{Z}\cdot\mathrm{diag}(\widebar{\bm{V}}^*)^{-1}\widebar{\bm{S}}.
\end{equation}
Numerical experiments show the approximation quality of the proposed method \eqref{eq:LPF} and FPL \eqref{eq:FPL} are incomparable when linearized around nominal operation point: FPL appears to be more accurate when loading levels are lower than the linearization point, and the proposed method has higher accuracy otherwise.

However, it can be shown that when both are linearized at zero injection point, the proposed method always yields smaller error than FPL when the underlying network is radial. This is stated in Proposition \ref{thm:LPFvFPL} below.
\begin{proposition} \label{thm:LPFvFPL}
    Consider both linear power flow models \eqref{eq:LPF0} and \eqref{eq:FPL} linearized around zero injection point. For given PQ bus power injections $\bm{S}$, let $\hat{\bm{v}}$ be the approximate squared voltage magnitudes given by \eqref{eq:LPF0} and $\tilde{\bm{v}}$ be the squared magnitude of the approximate voltage phasors given by \eqref{eq:FPL}. Let $\bm{V}$ be an actual power flow solution\footnote{There may be multiple power flow solutions for the given power injection $\bm{S}^*$, the result holds for any of them.} and denote $\bm{v} := \mathrm{diag}(\widebar{\bm{V}})\bm{V}$, then the inequality $\left| \hat{v}_i - v_i \right| \le \left| \tilde{v}_i - v_i \right|$ holds for all $i \in \mathcal{N}$ when the network is radial.
\end{proposition}
\begin{proof}
    It has been shown in the previous section that \eqref{eq:LPF0} coincides with the multiphase LinDistFlow model. So we know from \cite[Lemma 12-4]{SLowConvexRelaxOPF2014_partI} that $\hat{v}_i \ge v_i$ holds for all $i \in \mathcal{N}$ for a radial network. To prove the proposition it suffices to show $\tilde{v}_i \ge \hat{v}_i$ for all $i \in \mathcal{N}$.
    
    When linearized around zero injection point, $\bm{V}^*=\bm{E}$. Based on \eqref{eq:FPL}, $\tilde{\bm{v}}_i$ can be derived as,
\begin{align}
    \tilde{\bm{v}}
    &=
    \mathrm{diag}(\bar{\bm{E}})\bm{E} + \mathbf{M} \bm{p} + \mathbf{N} \bm{q} \nonumber\\
    &+ \mathrm{diag}(\mathbf{Z}\cdot \mathrm{diag}(\bar{\bm{E}})^{-1} \widebar{\bm{S}})^{H} (\mathbf{Z}\cdot \mathrm{diag}(\bar{\bm{E}})^{-1} \widebar{\bm{S}}) \nonumber\\
    &=
    \hat{\bm{v}}+ \mathrm{diag}(\mathbf{Z}\cdot \mathrm{diag}(\bar{\bm{E}})^{-1} \widebar{\bm{S}})^{H} (\mathbf{Z}\cdot\mathrm{diag}(\bar{\bm{E}})^{-1} \widebar{\bm{S}}).
\end{align}
Hence, $\tilde{v}_i \geq \hat{v}_i$ for all $i \in \cN$.
\end{proof}

\section{Numerical Evaluation}\label{subsec:numerical}
This section presents the numerical results of the proposed linearized power flow (the generalized LinDistFlow, or GLDF) model at various loading conditions, in comparison with the performance of LinDistFlow (LDF) model and the advanced FPL model in \cite{A.B2018} under the same conditions.


\subsection{Simulation Configuration}
We carry out our comparisons in IEEE 13-bus, 37-bus, and 123-bus radial distribution systems. For each system, we set its slack bus's voltage magnitude to 1 p.u. We then denote a reference constant power loading condition by  $\bm{S}^{\mathrm{ref}}$, under which the voltage magnitudes approximately range within [0.938, 1]~p.u., [0.946, 1]~p.u., and [0.908, 1]~p.u. for 13-bus, 37-bus, and 123-bus systems, respectively. We evaluate the approximation errors by continuation analysis and generating random load values. 

\subsection{Continuation Analysis}
We tune the testing loading conditions $\bm{S}=k\bm{S}^{\mathrm{ref}}$ by changing the parameter $k$ from -2.5 to 2.5 with granularity of $0.01$. For each system, we present the average estimation error as a function of system loading levels by each of the three methods (GLDF, LDF, and FPL) linearized at three operation points: zero injection point with $k=0$, a positive rated loading point with $k=1$, and a negative rated loading point with $k=-1$. Note that LDF always has a fixed linearization formula regardless of the loading conditions. We record the linearization errors of GLDF, LDF, and FPL by comparing their linearized voltage magnitudes against the true values from the nonlinear power flow model. Here, we use the metrics of relative errors defined as $\| |\tilde{\bm{V}}|-|\bm{V}|\|_2/\| |\bm{V}|\|_2$ with the vector of true voltage magnitudes $|\bm{V}|$ and the vector of linearized voltage magnitudes $|\tilde{\bm{V}}|$.
\subsubsection{Compare GLDF and LDF}
We have analytically characterized the relationship between the traditional LDF and the proposed GLDF in Section~\ref{sec:LDF}: LDF is a special example of GLDF linearized at zero injection point. This has been validated in the left figures of Fig.~\ref{fig:IEEE13}--\ref{fig:IEEE123}, where we show that the error rates of LDF (green dashed lines) coincide with those of GLDF (blue solid lines) when GLDF is linearized at zero injection point for all testing systems.

However, because LDF is designed to be a fixed linearization method under the assumption of lossless power flow---or often times equivalently, little nodal power injections---its performance declines as the operating points deviate from the zero injection. In such scenarios, we may benefit from using localized linearization method to achieve better accuracy using GLDF or FPL. As shown in the middle (linearized at $k=1$) and right figures (linearized at $k=-1$) of Figs.~\ref{fig:IEEE13}--\ref{fig:IEEE123}, both GLDF and FPL perform much better than LDF around their linearized points at $k=1$ and $k=-1$.

\subsubsection{Compare GLDF and FPL}
Both GLDF and FPL can generate linearization based on centain operating points. When linearized at zero load point, GLDF is equivalent to LDF, and always performs better than FPL; see the left figures of Figs.~\ref{fig:IEEE13}--\ref{fig:IEEE123}. This results also echoes with the error analysis of Proposition~\ref{thm:LPFvFPL}. 

On the other hand, when linearized at $k=1$ (resp. $k=-1$), GLDF generates the same performance as FPL at the linearized points. While FPL has better accuracy from around $k=0$ to $k=1$ (resp. $k=-1$ to $k=0$) because it is an approximation of an exact two-point linearization based on $k=0$ and $k=1$ (resp. $k=-1$), GLDF generates more consistent results beyond $k<0$ and $k>1$ (resp. $k>0$ and $k<-1$).

\subsection{Evaluation of Random Load}
When we generate random load values, we consider a system without and with the penetration of DER. For both cases, we use the metrics of average errors defined as $\frac{1}{nk} \sum_{j=1}^{k} \sum_{i=1}^{n} |(|\tilde{\bm{V}}_i^j|-|\bm{V}_i^j|)|$ and maximum errors defined as $\underset{j} \max  \||\tilde{\bm{V}}^j|-|\bm{V}^j|\|_{\infty}$, where $\tilde{\bm{V}}_i^j$, $\bm{V}_i^j$ denote voltage magnitude at node $i$ for the $j$th sample given by LPF models and exact power flow model, respectively, and $k=10,000$ denotes the number of samples.

On the first case, random load values and power factor follow uniform distribution, denoted by $p_i \sim U(1.5 \re (\bm{S}^{\mathrm{ref}}_i), 0)$, $pf_i \sim U(0.7,1)$, where $U$ denotes uniform distribution, $p_i$, $\bm{S}^{\mathrm{ref}}_i$, $pf_i$ denote real power injection, reference power and power factor at node $i$. We calculate the errors of the three methods linearized at $0.75\bm{S}^{\mathrm{ref}}$ and present the mean and maximum values in TABLE \ref{RatedLoad}.

 For the second case, power factor follows uniform distribution, denoted by $pf_i \sim U(0.7,1)$. Random load values for half of load nodes follow uniform distribution, denoted by $p_i \sim U(1.5 \re (\bm{S}_i), 0)$, where $\bm{S}_i = \bm{S}^{\mathrm{ref}}_i$, while those of the other load nodes follow uniform distribution, denoted by $p_i \sim U(0, 1.5 \re (\bm{S}_i))$, where $\bm{S}_i =- \re (\bm{S}^{\mathrm{ref}}_i) + \mathfrak{i} \im (\bm{S}^{\mathrm{ref}}_i)$. We calculate the errors of the three methods linearized at $0.75\bm{S}$ and summarize them in TABLE \ref{Negative RatedLoad}.

As shown in TABLE \ref{RatedLoad}--\ref{Negative RatedLoad}, the mean and maximum errors of GLDF is smaller than FPL. We get similar results when linearized at other operating points. Therefore, GLDF is more robust to variation of load values than FPL method.

\begin{figure*}
\centering
\includegraphics[width=.32\textwidth]{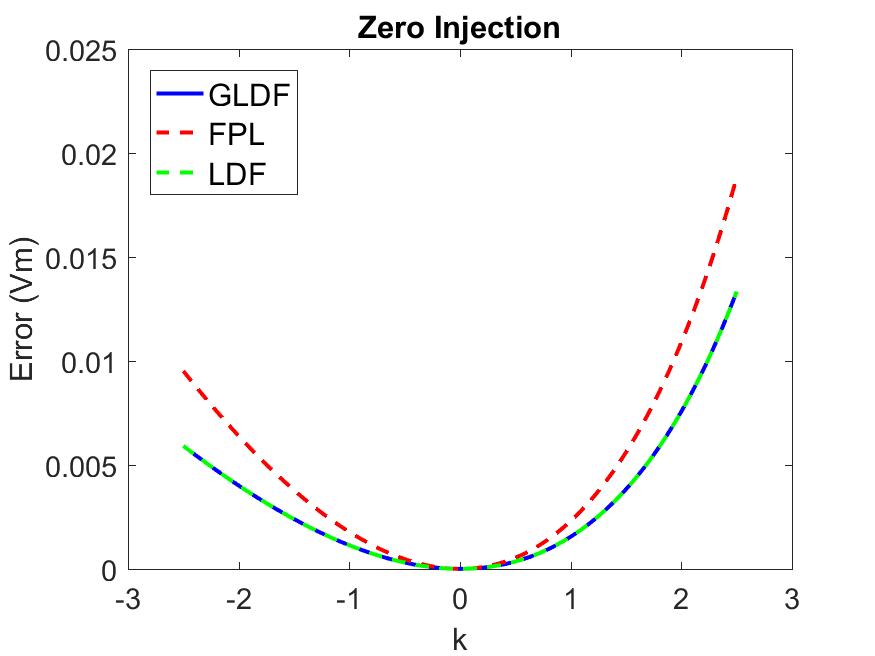}
\includegraphics[width=.32\textwidth]{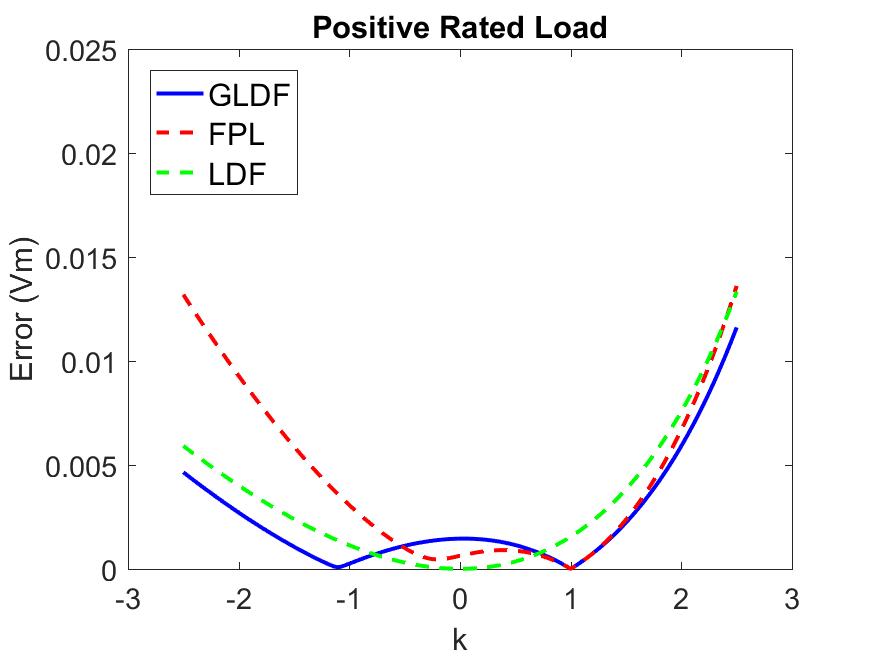}
\includegraphics[width=.32\textwidth]{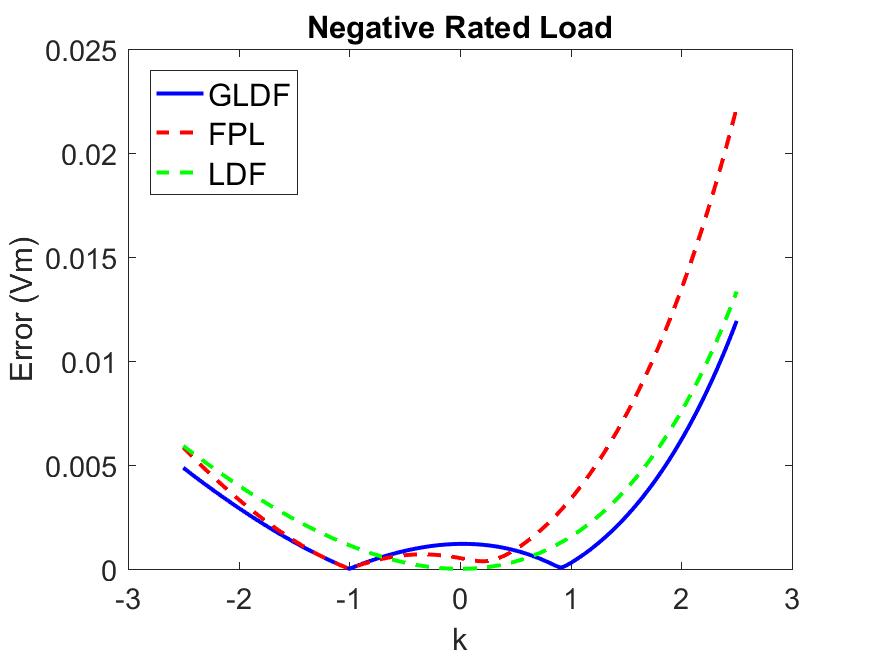} 
\caption{Relative linearization errors of voltage magnitudes in IEEE 13-bus feeder when linearized at (left) zero load, (middle) positive rated load, and (right) negative rated load.}
\label{fig:IEEE13}
\end{figure*}

\begin{figure*}
\centering
\includegraphics[width=.32\textwidth]{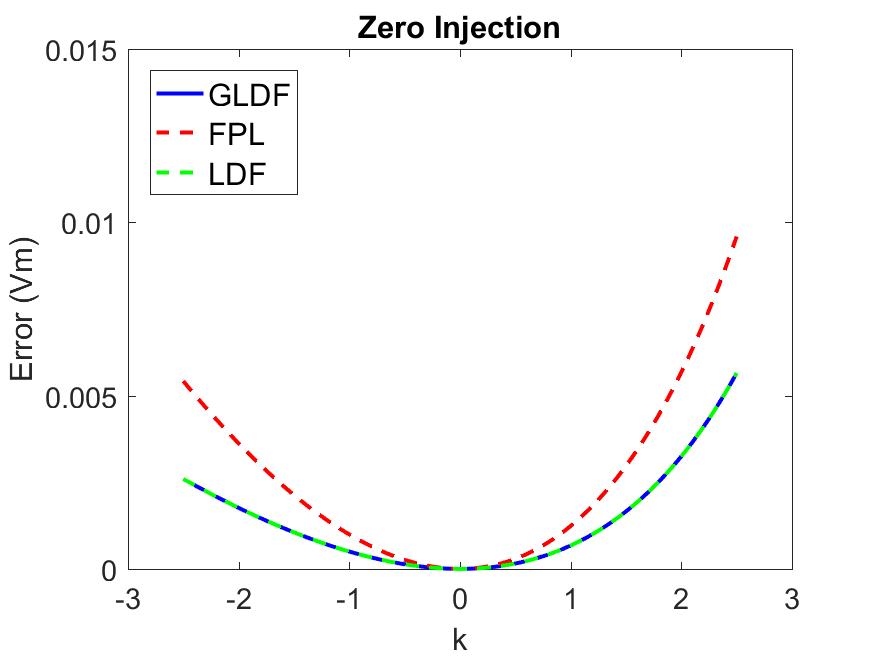}
\includegraphics[width=.32\textwidth]{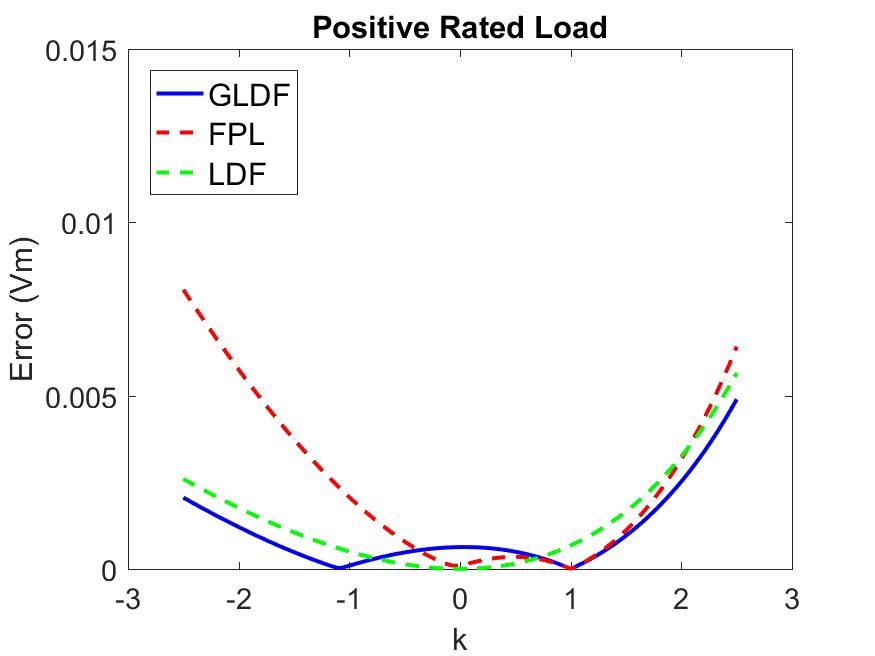}
\includegraphics[width=.32\textwidth]{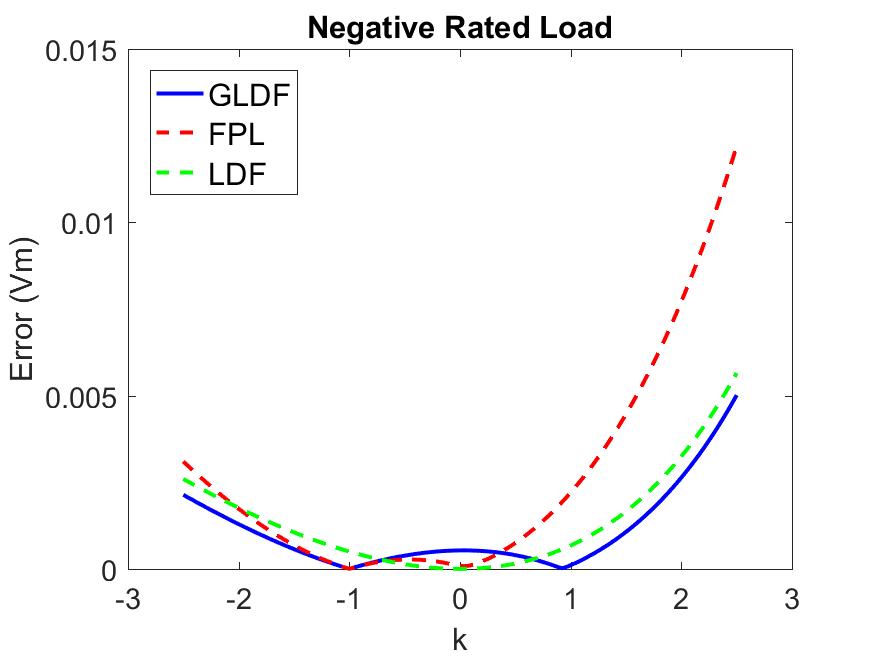} 
\caption{Relative linearization errors of voltage magnitudes in IEEE 37-bus feeder when linearized at (left) zero load, (middle) positive rated load, and (right) negative rated load.}
\label{fig:IEEE37}
\end{figure*}

\begin{figure*}
\centering
\includegraphics[width=.32\textwidth]{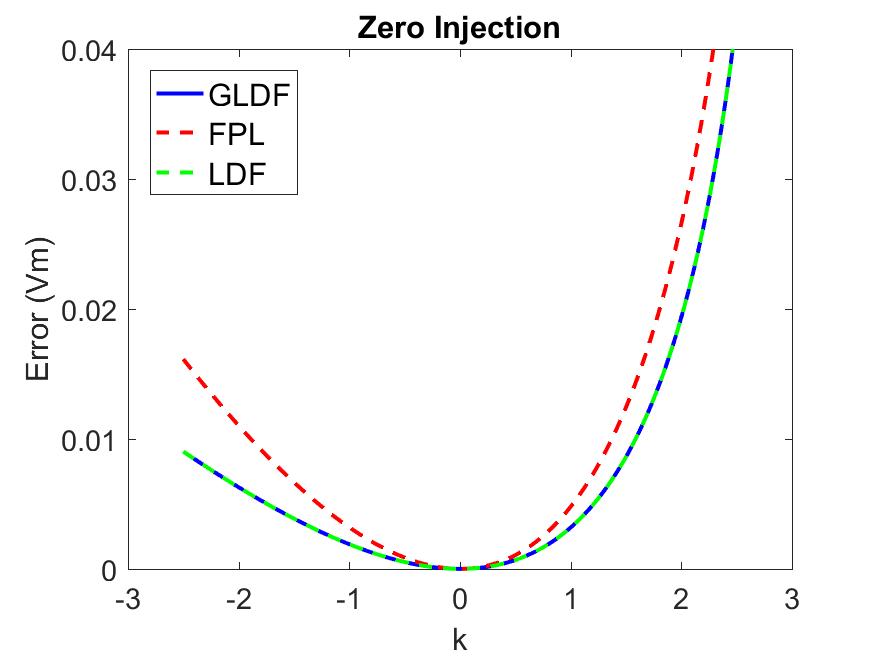}
\includegraphics[width=.32\textwidth]{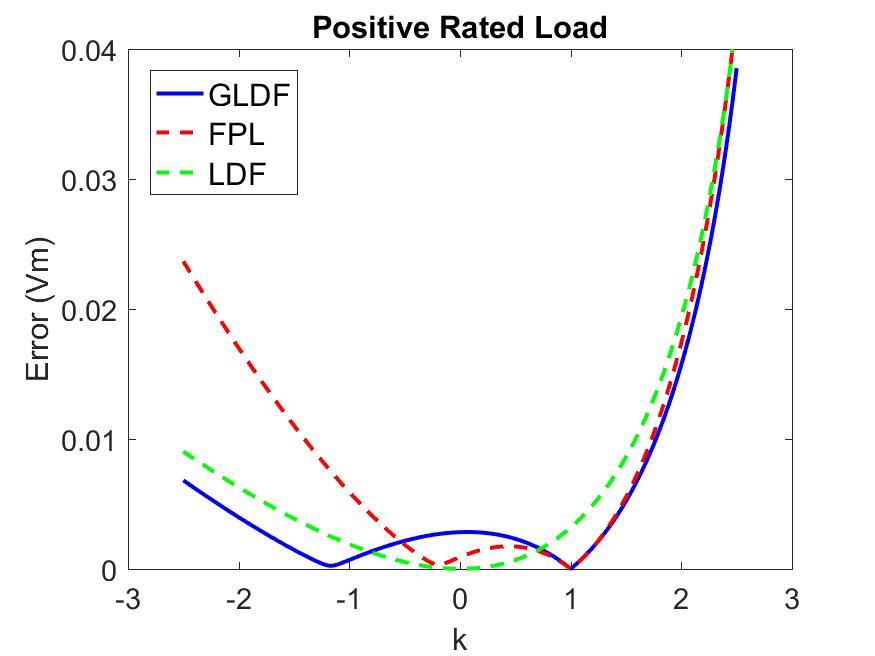}
\includegraphics[width=.32\textwidth]{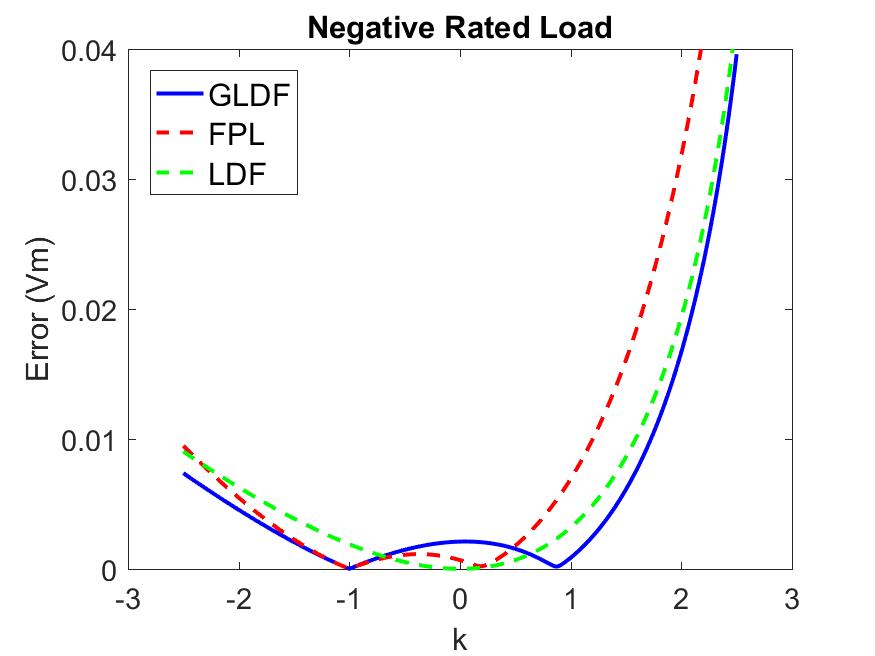} 
\caption{Relative linearization errors of voltage magnitudes in IEEE 123-bus feeder when linearized at (left) zero load, (middle) positive rated load, and (right) negative rated load.}
\label{fig:IEEE123}
\end{figure*}

\begin{table}[htbp]
\begin{center}
\caption{Approximation Errors of voltage magnitudes for positive load. The unit is 0.01 p.u.}
\begin{tabular}{c c c c c c c}
\dtoprule
\multicolumn{1}{c}{\textbf{IEEE}} 
& \multicolumn{2}{c}{\textbf{LDF}}
& \multicolumn{2}{c}{\textbf{GLDF}} 
& \multicolumn{2}{c}{\textbf{FPL}} \\
\cline{2-7} 
\textbf{Feeder} & \text{Mean}& \text{Max}& \text{Mean}& \text{Max}& \text{Mean}& \text{Max} \\
\hline
13  & $0.103$  & $1.62$    & $0.0855$  & $1.48$ & $0.0941$ & $2.35$ \\
37  & $0.0356$  & $0.295$  & $0.0168$  & $0.214$ & $0.0299$ & $0.423$ \\
123 & $0.143$  & $0.938$  & $0.0644$  & $0.704$ & $0.103$ & $1.14$ \\
\dbottomrule
\end{tabular}
\label{RatedLoad}
\end{center}

\begin{center}
\caption{Approximation Errors of voltage magnitudes for feeders with the penetration of DER. The unit is 0.01 p.u.}
\begin{tabular}{c c c c c c c}
\dtoprule
\multicolumn{1}{c}{\textbf{IEEE}} 
& \multicolumn{2}{c}{\textbf{LDF}}
& \multicolumn{2}{c}{\textbf{GLDF}} 
& \multicolumn{2}{c}{\textbf{FPL}} \\
\cline{2-7} 
\textbf{Feeder} & \text{Mean}& \text{Max}& \text{Mean}& \text{Max}& \text{Mean}& \text{Max} \\
\hline       
13  & $0.0932$  & $1.32$   & $0.0755$  & $1.15$ & $0.0799$ & $1.93$ \\
37  & $0.013$  & $0.120$  & $0.00839$  & $0.106$ & $0.0105$ & $0.205$ \\
123 & $0.0397$  & $0.333$  & $0.0236$  & $0.263$ & $0.0326$ & $0.490$ \\
\dbottomrule
\end{tabular}
\label{Negative RatedLoad}
\end{center}

\end{table}

\section{Conclusion and Future Work}\label{sec: conclusion}
In this paper, we propose a novel lineaerized power flow model for distribution system. The proposed model can be applied to a multiphase power grid with generic topology, like radial or meshed networks, under balance or unbalance three phase voltage. We show that the proposed model delivers more accurate voltage than FPL method when they are linearized at zero injection point. Numerical experiments show that for other linearization points, the proposed model is more robust to load variation than FPL. Comparing with FPL, the approximation errors of the proposed model slowly increase even when the exact operating point is moving far away from the linearization point.  

In the future, we will analyze the approximation errors of the proposed model theoretically when it is linearized at other operating points, like rated load, to further improve the accuracy and robustness of the model. We will also apply the proposed model to power system operation problem, like real-time power flow, OPF and SE problems to see whether we can obtain better solution with the proposed model.

\bibliographystyle{IEEEtran}
\bibliography{GLDF.bib}

\end{document}